\documentclass[runningheads]{llncs}

\usepackage{balance} 
\usepackage[T1]{fontenc}
\usepackage{graphicx}
\usepackage{amsfonts} 
\usepackage{amsmath} 
\usepackage{todonotes} 
\usepackage{xspace} 
\usepackage{subcaption} 
\usepackage[thinc]{esdiff}
\usepackage{hyperref}  
\newcommand{\naturals}{\ensuremath{\mathbb{N}}\xspace}

\long\def\remove#1{}
\newcommand{\timeout}{judiciary period\xspace}
\newcommand{\gamename}{Queue\xspace}
\newcommand{\roundname}{Round\xspace}

\newcommand{\nashconv}{{\rm NashConv}\xspace}
\newcommand{\ifinesname}[1]{\ensuremath{#1}-Fines\xspace}
\newcommand{\icritstrategy}[1]{\ensuremath{\pi_a^{{\rm crit}, #1}}\xspace}
\newcommand{\critstrategy}{\ensuremath{\pi_a^{{\rm crit}}}\xspace}
\newcommand{\critstrategyext}[1]{\ensuremath{\pi_{#1}^{{\rm crit}}}\xspace}

\newcommand{\agentset}{\ensuremath{\mathcal{N}}\xspace}

\newcommand{\terminalset}{\ensuremath{\mathcal{T}}\xspace}
\newcommand{\game}{\ensuremath{\mathbb{G}}\xspace}
\newcommand{\gamefull}{\ensuremath{\mathbb{G}(F,Q,T,k,p,x,x_0,w)}\xspace}
\newcommand{\round}{\ensuremath{\mathbb{O}}\xspace}
\newcommand{\roundfull}{\ensuremath{\mathbb{O}(\agentset,F,Q,T,k,p)}\xspace}

\newcommand{\ifinesfull}[1]{\ensuremath{\mathbb{F}(#1,F,Q,k,p,x_0)}\xspace}
\newcommand{\ifinesfullk}[2]{\ensuremath{\mathbb{F}(#1,F,Q,#2,p,x_0)}\xspace}
\newcommand{\argmax}{{\rm argmax}\xspace}

\begin{document}

\title{Rule Enforcing Through Ordering
\thanks{This work has been supported by the CoSP, project n. 823748 H202-MSCA-RISE-2018. Computational resources were supplied by the project e-Infrastruktura CZ (e-INFRA CZ LM2018140) supported by the Ministry of Education, Youth and Sports of the Czech Republic.}}
%
%
\author{David Sychrovsky\inst{1}\orcidID{0000-0002-4826-1096} \and
Sameer Desai\inst{2}\orcidID{0000-0003-1987-6929} \and
Martin Loebl\inst{1}\orcidID{0000-0001-7968-0376}}
\authorrunning{D. Sychrovsky et al.}
%
\institute{$^1$Charles University, $^2$University of Passau\\
\email{\{sychrovsky, loebl\}@kam.mff.cuni.cz,
sameer.indirock@gmail.com}}
\maketitle              
\begin{abstract}

In many real world situations, like minor traffic offenses in big cities, a central authority is tasked with periodic administering punishments to a large number of individuals. Common practice is to give each individual a chance to suffer a smaller {\em fine} and be guaranteed to avoid the legal process with probable considerably larger punishment. 
However, thanks to the large number of offenders and a limited capacity of the central authority, the individual risk is typically small and a rational individual will {\it not} choose to pay the fine. 
Here we show that if the central authority processes the offenders in a publicly known order, it properly incentives the offenders to pay the fine. 
We show analytically and on realistic experiments that our mechanism promotes non-cooperation and incentives individuals to pay. Moreover, the same holds for an arbitrary coalition. We quantify the expected total payment the central authority receives, and show it increases considerably.

\end{abstract}





         
\newcommand{\BibTeX}{\rm B\kern-.05em{\sc i\kern-.025em b}\kern-.08em\TeX}




\section{Introduction}
\label{sec: intro}
In this work, we study a special case of a classic dilemma, how to effectively enforce a rule
in a large population with only a very small number of enforcing agents. This task is impossible if the large population cooperates and thus a critical aspect of any suggested mechanism is the promotion of non-cooperation. A well-known count Dracula way is to make the punishment for breaking the rule extremely severe. We suggest an alternative mechanism, for a special case of the dilemma motivated by collecting fines for traffic violations. 

In many large cities, there is a huge number of traffic offences, highly exceeding the capacity of state employees assigned to manage them. 
The assigned state employees should primarily concentrate on serious and repetitive offenders. However, a large number of minor offences are still to be settled which makes the former considerably harder. A common practise is that a smaller {\em fine} is assigned in an almost automated way and if an offender settles this fine then the legal process does not start. Otherwise, the legal process should start with considerably larger cost for the offender. The offence is also forgotten after a certain {\it \timeout}.  

However, thanks to the limited capacity of state employees, legal processes for non-repetitive minor traffic offenses are typically enforced in a small number of cases\footnote{For instance, in the city of Prague considerably more than 100 000 such offenses are dismissed every year because the \timeout expires.}. The individual risk is thus small and a large fraction of the offenders {\it choose} to ignore the fine. In this paper, we propose a simple mechanism which properly incentives the offenders to pay the fine even under these conditions.

\subsection{Main Contribution}
\label{ssec: contribution}

In our proposed mechanism, the central authority processes the offenders in a given order. Each offender is aware of his position in this `queue of offenders' and has the option of publicly donating money to a fund of traffic infrastructure or a charity predetermined by the central authority. If their total donations amount to at least the fine, it is used to settle the offence. After the \timeout expires, or if the legal process is started, the fund retains the individual donation.
The central authority periodically sorts offenders in ascending order of their average donation, and starts the legal process with those who paid the least on average.

Compared to processing the offenders in random order, this mechanism increases the individual risk of some offenders. This incentives them to pay the fine, which in turn puts others in danger.
We show both analytically and on realistic experiments that under the proposed mechanism, the strategic behaviour of the offenders is to engage with the mechanism, and quantify the expected revenue of the charity. 
Moreover, we show it is not beneficial for any group of offenders to ignore the mechanism and share the cost of those who enter the legal process. 
Finally, we study how the central authority can most efficiently use its limited capacity to maximize the revenue of the charity.

This paper is a continuation of \cite{Sychrovsky2023Promoting}, where the authors introduced the model studied here. We extend their work by providing a complete solution to $w$-Fines, see Section~\ref{sec: analytic}, as well as producing more thorough numerical experiments.

\subsection{Related Work}
\label{ssec: related work}

To our best knowledge, the field of non-cooperative mechanism design has not been studied extensively yet. Our approach is somewhat similar to that of \cite{Bidding2020Chiam}, where the authors consider a variation of the elimination game which includes bids. 
Our model can also be viewed as a generalization of the stopping games \cite{Kobylanski2013Dynkin}, where participants choose a time to stop bidding and trade off their gain from outlasting other players for the cost accumulated over time in the game. In our case, the ``prize'' won by the lowest paying participant is cost of entering the legal process.
However, both approaches did not consider the ranking of players, which is at the core of our mechanism.
\section{Problem Definition}
\label{sec: problem definition}

Informally, we model the interaction of agents as a game we call {\it \gamename}. \gamename consists of a finite sequence of {\it \roundname}, in which each agent can choose to pay, however with some probability they forget and pay nothing. Those who paid at least the fine in total, or spent enough time in \gamename are removed. The rest is ordered according to the amount they paid on average. A fixed number of those at the start are then forced to pay a large penalty, and leave \gamename.
Let us now define the interaction formally, starting with how \roundname is realized.

\subsection{\roundname: One Step in \gamename}
\label{ssec: round}
\roundname is a parametric game \round$(\agentset)$=\roundfull, where $\agentset$ is an ordered subset of agents\footnote{The agents are ordered according to their average payment in ascending order, i.e. those who paid the least on average are sorted to the front of \agentset.}, $F \in \naturals$ is the fine, $Q > F$ is the cost associated with entering the legal process, $T\in\naturals$ is the \timeout, i.e., the number of \roundname instances after which agents are removed, $k\in\naturals$ is the number of agents forced to pay $Q$ in each \roundname, $p\in[0,1]$ is the probability of ignorance.

Each $a\in \agentset$ is characterized by a triplet $(n_a, t_a, m_a)$ and his strategy $\pi_a$. The triplet corresponds to his {\it observations} --- his position $n_a$ in $\agentset$, the number $t_a$ of past \roundname games he participated in\footnote{This includes the current Round, i.e. $t_a\ge 1$.}, and his total individual payment $m_a$ in the past \roundname games. 

\roundname proceeds in three phases
\begin{enumerate}
    \item Each agent $a\in \agentset$, based on his observation, declares his strategy for this \roundname $\pi_a\in\Delta^{F+1}$, where $\Delta$ is the probability simplex. His payment $\mu_a$ is then sampled from\footnote{This simulates that with probability $p$, the agent forgot to act in this \roundname.}
    \begin{equation}
        \mu_a \sim p\sigma^0 + (1-p)\pi_a(n_a, t_a, m_a),
    \end{equation}
    where $\sigma^\nu$ is the pure strategy of paying $\nu$.
    \item Each agent's total payment and time is updated
    \begin{align}
        m_a &\gets m_a + \mu_a,\\
        t_a &\gets t_a + 1,
    \end{align}
    and $\agentset$ is sorted\footnote{We use stable sort, i.e. whenever there is a tie, the original order is preserved.} according to the ratio of current total payment and time $m_a / t_a$.
    \item Some agents are removed from \agentset, which is done in three sub-phases. We call such agents {\it terminal} and denote the set of terminal agents in this \roundname as \terminalset.
    \begin{enumerate}
        \item All agents $a\in \agentset$ with $m_a \ge F$ are removed.
        \item First $k$ agents in $\agentset$ have their $m_a$ increased by $Q$ and are removed.
        \item All agents $a\in \agentset$ with $t_a \ge T$ are removed.
    \end{enumerate}
\end{enumerate}

The result of each \roundname is the ordered set of agents $\agentset\setminus\terminalset$, and the set of terminal agents \terminalset. Only the terminal agents are assigned their final utility.
\begin{definition}[Utility]
The utility of each agent $a\in \terminalset$ is the negative amount he paid
\begin{equation}
\label{eq: utility}
    u_a = -m_a.
\end{equation}
\end{definition}

\subsection{\gamename: A Game on Updating Sequences} 
\label{ssec: game}
Formally, \gamename is \game = \gamefull, where $F,Q,T,k$ and $p$ have the same meaning as in Section~\ref{ssec: round}, $x$ is the number of entering agents after each \roundname, $x_0$ is the initial size of $\agentset$ and $w$ is the horizon, i.e. the number of repetitions of \roundname.

\gamename aggregates \roundname in the following two simple phases. Starting with $\agentset^1$ s.t. $|\agentset^1|=x_0$, and $m_a,t_a=1$ for each $a\in\agentset^1$. We repeat them $w$-times.
\begin{enumerate}
    \item The agents in $\agentset^t$ play \roundname and non-terminal agents proceed to the next iteration.
    \begin{equation}
        \agentset^{t+1}, \terminalset^{t+1} \gets \round(\agentset^t).
    \end{equation}
    \item $x$ new agents enter the game
    \begin{equation}
        \agentset^{t+1} \gets \agentset^{t} \cup X,
    \end{equation}
    where $X$ is a set of agents with $m_a, t_a = 0$, and $|X|=x$. These new agents are sorted to the end of $\agentset^{t+1}$.
\end{enumerate}
\*
In the last \roundname, all agents terminate, $\terminalset^w \gets \terminalset^w \cup \agentset^w$.

\medskip

The new agents come from universum $U$. The strategy of all agents is then given as $\pi=\times_{a\in U}\pi_a$. We denote space of all such strategies as $\Pi$.

Each agents wants to choose strategy $\pi_a$, which maximizes their utility in \game given strategies of other agents $\pi_{-a}$. A strategy profile $\pi\in\Pi$ is an equilibrium, if no agent can increase his utility. Formally,
\begin{definition}[$\epsilon$-Equilibrium]
\label{def: equilibrium}
$\pi\in\Pi$ is an $\epsilon$-equilibrium of \game if $\forall\ \overline{\pi}\in\Pi$, $\forall t\in \{1,\dots w\}$ and $\forall a\in\terminalset^t$,
\begin{equation}
    \mathbb{E}_\pi\left[u_a(\pi)\right] \ge 
    \mathbb{E}_{(\overline{\pi}_a \pi_{-a})}\left[u_a(\overline{\pi}_a, \pi_{-a})\right] - \epsilon.
    \hspace{3ex}
\end{equation}
\end{definition}

We note that the equilibrium always exists which can be shown by a standard transformation to a normal form game. 

\subsection{Avalanche Effect}
\label{ssec: avelanche effect}
Intuitively, every agent wants to pay as little as possible, while avoiding paying $Q$. This translate to paying more than the others. However, if all agents adapt this reasoning, the only option to avoid paying $Q$ is to pay $F$. We formally show this in Section \ref{ssec: p=0}

Crucially, not all other agents can use this reasoning thanks to the probability of ignorance. But as that vanishes, the agents should be incentivised to pay more. Similarly, if the number of entering agents increases, so should the total payment.
We formally capture this in the {\it avalanche effect}.
\begin{definition}[Avalanche effect]
\label{def: avalanche effect}
We say that \gamename exhibits the {\em avalanche effect} if at least one of the following holds in equilibrium when changing $p$, or $x$.
\begin{enumerate}
    \item 
    The expected terminal payment of all agents is increasing with $p\to 0^+$
    \begin{equation}
        \lim_{p\to0^+}\frac{\rm d}{{\rm d} p}\sum_{a\in\terminalset}m_a < 0.
    \end{equation}
    \item
    The expected terminal payment of all agents decreases slower than $1 / x$
    \begin{equation}
        \frac{\rm d}{{\rm d} x}\sum_{a\in\terminalset}m_a > 0, \hspace{3ex} \forall x > 0.
    \end{equation}
    
\end{enumerate}
\end{definition}

\subsection{Division Problem}
\label{ssec: division problem}
In our model, the \timeout is split into $T$ equal time intervals and sorted at the start of each interval. The central authority can process $kT$ offenders over the \timeout, and $xT$ will enter the system.

The central authority can influence the system in two ways. 
\begin{enumerate}
    \item it can choose how often the sorting takes place, and
    \item it can virtually split the entering offenders into $g$ groups of size $x/g$, and process $k/g$ offenders in each.
\end{enumerate}

The {\it Division problem} is how to set $T$ and $g$ to maximize the expected revenue the central authority receives. We refer to the two cases as {\it Time-Division problem} and {\it Group-Division problem} respectively.

\section{Analytic Solution}
\label{sec: analytic}

As described in Section~\ref{sec: intro}, the individual risk when the central authority processes the agents in random order is typically small, i.e. $kQ/|\agentset| \ll F$. Each agent is also guaranteed to pay $kQ/|\agentset|$ if everyone cooperates and shares the costs of those entering the legal process. Let us begin by showing that this is not the case in our proposed system. That is, there is no coalition can benefit from choosing to pay nothing and share the cost of those forced to pay $Q$. In our setting, this is analogous to coalition proofness.
\begin{proposition}
Let $\mathcal{A}$ be a set of agents using strategy $\pi_a = \sigma^0\ \forall a \in \mathcal{A}$, and sharing the cost, i.e. their utility becomes
\begin{equation*}
    \tilde{u}_a = -\frac{1}{|\mathcal{A}|}\sum_{i=1}^w\sum_{a\in \mathcal{A}\cap \terminalset^i} m_a,
    \hspace{3ex}
    \forall a \in \mathcal{A}.
\end{equation*}
If $\tilde{u}_a<0$, then $\exists a' \in \mathcal{A}$ s.t. $a'$ can deviate and increase his utility. 
\end{proposition}
\begin{proof}
We split the proof into two parts according to how much an individual needs to contribute.
\begin{enumerate}
    \item $0>\tilde{u}_a > Q$: In this situation, not all agents of $\mathcal{A}$ were forced to pay $Q$. Consider the agent $a'\in\agentset$ who terminated last. Then, since $a'$ paid zero, his original utility is zero and $\tilde{u}_a<u_a$. Therefore, $a'$ would benefit from leaving the coalition $\mathcal{A}$.
    \item $\tilde{u}_a=-Q$: In this case, all agents were forced to enter the legal process. Any $a\in\mathcal{A}$ would therefore benefit from paying the fine, since then his utility is $u_a = -F > -Q = \tilde{u}_a$.
\end{enumerate}
\end{proof}

While existence of an analytic solution of \gamename remains an open question, we can find it in certain special cases. 

\subsection{Active participants}
\label{ssec: p=0}
Let us first focus on a situation when no agent forgets to participate in \roundname, i.e. $p=0$.
Then it is easy to see that $\pi_a = \sigma^F$ is unique equilibrium. 
Consider the first agent $a\in \agentset$ in the first \roundname, who chose to pay $\mu_a < F$. Then he is forced to pay $Q$, resulting to utility $u_a = -Q-\mu_a < -F$. Therefore, switching to paying $F$ is beneficial and the strategy of paying $\mu_a < F$ is not an equilibrium. This means all agents will pay $F$ in the first \roundname, and the situation thus repeat in the following \roundname.

\subsection{\ifinesname{w}: Special Case of \gamename}
Let us focus on the system without the introduction of the option to donate a portion of the fine. Thus after scaling currency we can let $F=1$, and there are only two pure strategies $\sigma^0, \sigma^F$ the agents can take.
If now $T=w$ and no agents are added after each \roundname $x=0$, we call the game {\it \ifinesname{w}}.

\begin{definition}[\ifinesname{w}]
We refer to reduced \gamename $$\ifinesfull{w} = \game(F,Q,w,k, p,0,x_0,w)$$ as \ifinesname{w}.
\end{definition}

We begin by showing a crucial property of \ifinesname{w}.

\begin{lemma}
\label{lem: irelevant after}
In the \ifinesname{w}, the expected payment of $\forall a\in\agentset$ depends only on the actions of agents in front of $a$.
\end{lemma}
\begin{proof}
If $a$ pays zero, he remains in the \gamename and is sorted in front of agents who were behind him. He is potentially forced to pay $Q$, depending on the actions of agents in front of him. If he pays $F=1$, he is removed. In either case, the actions of agents behind $a$ have no impact on his payment.
\qed
\end{proof}

In each \roundname, $a\in\agentset$ has $n_a-1$ agents in front of him. Due to the probability of ignorance, even if all the agents decide to pay, $a$ can estimate the probability that at most $k-1$ will forget. If that happens, $a$ will be forced to pay $Q$ in this \roundname. Formally,

\begin{definition}
\label{def.a}
Let $n$ be a positive integer. We denote by $\alpha(p,n,k)$ the probability that in $n-1$ independent coin tosses with the head probability $p$, the number of heads is less than $k$.
\end{definition}

Since $\alpha$ will be important in the following discussion, we briefly mention some of its properties.

\begin{lemma}
\label{l.al}
Let $k<np$, then $\alpha(p,n+1,k) \leq  e^{-\frac{( np-k )^2}{2np}}$.
\end{lemma}
\begin{proof}
Let $\xi_i$ denote the random variable such that $$ \xi_i = \begin{cases}
  1  &  \text{ w.p. } p, \\
  0 &   \text{ otherwise },
\end{cases}
$$
and $\xi^n = \underset{i=1}{\overset{n}{\sum}}\xi_i$. Thus, $\mathbb{E}[\xi_i]=p$ and $\mathbb{E}[\xi^n]=np$.
As per the Chernoff bounds, $\mathbb{P}[\xi^n \leq (1-\delta)np]\leq e^{\frac{-\delta^2np}{2}}$, for all $ 0<\delta<1$.
Thus $\alpha(p,n+1,k)=\mathbb{P}[\xi^n \leq k] \leq e^{{-\left( 1-\frac{k}{np} \right) ^2np}\slash{2}} = e^{-\frac{\left( np-k \right)^2}{2np}}$.
\end{proof}

\begin{proposition}
\label{p.aa}
If $\alpha(p,n,k)\leq F/Q\leq \frac{1}{4}$ then $np > k$. Moreover for each positive integer $w$ and large enough $n$, $\alpha(p,n,k)\geq \alpha(p,wn,wk)$.
\end{proposition}
\begin{proof}

For $\gamma\sim B(n,p)$ if $ p<1-\frac{1}{n}$, then $\frac{1}{4}<{\rm Pr}(\gamma\leq np)$~\cite{GREENBERG201491}. Therefore, when $\frac{1}{4}\geq \frac{F}{Q}$, then $k<np$.
Further, we note that Lemma \ref{l.al} is tight for large enough $np$. Hence, it suffices to prove the proposition for the upper bound 
$e^{-\frac{\left( np-k \right)^2}{2np}}$ for which the statement clearly holds.
\end{proof}

Finally, we report a result that strengthens the second part of Proposition \ref{p.aa} for $w= 2$. 


\begin{theorem}
\label{c.aa}
$\alpha(p,n,k)\geq \alpha(p,2n,2k)$ for $1 \leq k<np-p$. 
\end{theorem}
The proof can be found in Appendix \ref{app: sameer proof}.

\subsubsection{Single Sorting Instance}
\label{sub.g}

We start by analysing the \ifinesname{1} game, which is equivalent to one \roundname. In this case, when an agent is sufficiently far from the start of \agentset, it is beneficial to pay nothing, while near the start it is beneficial to pay and avoid paying $Q$. The boundary between the two will prove important.

\begin{definition}[Critical strategy]
\label{def.crit}
Let $r>0$ be the smallest integer such that $\alpha(p,r,k)Q\leq F$. Then $r$ is called {\em critical position}.

The {\it critical strategy} is
\begin{equation}
\label{eq: 1-crit strategy}
    \critstrategy(n_a, t_a, m_a) = 
    \begin{cases}
    \sigma^F & {\rm if}\ \alpha(p,n_a,k)Q > F,\\
    \sigma^0 & {\rm otherwise}.
    \end{cases}
\end{equation}
\end{definition}
We note that $t_a=1$ and $m_a=0$ $\forall a\in\agentset$ for \ifinesname{1}. We will show that \critstrategy is the only equilibrium of the \ifinesname{1}. First, we define $\alpha^{\rm crit}$ as the probability with which an agent is forced to pay $Q$ when all agents follow \critstrategy.
\begin{proposition}
\label{p.crt}
Let $r$ be the critical position. Then if all agents but $a$ follow \critstrategyext{b}, and $a$ uses $\sigma^0$, then $a$ is forced to pay $Q$ w.p.
\begin{equation}
    \alpha^{\rm crit}(p, r, n_a, k) =
    \begin{cases}
    \alpha(p, n_a, k) & {\rm if }\ n_a < r,\\
    \alpha(p, r, k-(n_a-r)) & {\rm otherwise}.
    \end{cases}
\end{equation}
\end{proposition}
\begin{proof}
Fix $a\in\agentset$. When $\alpha(p,n_a,k) > F/Q$ (i.e. $n_a < r$), then agents in front of $a$ pay $F$ and thus $a$ will not pay $Q$ only if enough of them forget. If $n_a \ge r$, then $n_a-r$ agents choose not to pay. Therefore, $a$ only needs $k-(n_a-r)$ of the $r$ agents to forget.
\qed
\end{proof}
Observe that $\alpha^{\rm crit} \le \alpha$, since some agents may choose to pay zero. Also, by Definition~\ref{def.a}, $\alpha^{\rm crit} = 0$ for $n_a > r+k$.

\begin{proposition}
\label{thm.p}
Let $r$ be the critical position and let all agents follow \critstrategy, except for $a\in\agentset$, whose strategy is $\pi_a = (q, 1-q)$. Then the expected payment of $a$ is
\begin{equation}
\label{eq: exp 1 payment}
    (1-p-q)F + (p+q)\alpha^{\rm crit}(p,r,n_a,k)Q.
\end{equation}
\end{proposition}
\begin{proof}
By definition of $\pi_a$, $a$ pays $F$ w.p. $1-p-q$ and he does not forget. If he does, or pays zero w.p. $q$, he is forced to pay $Q$ w.p. $\alpha^{\rm crit}(p,r,n_a,k)$. 
\qed
\end{proof}

\begin{corollary}
\label{cor: exp payment}
Let $r$ be the critical position and let all agents follow \critstrategy. Then the expected payment of $a\in\agentset$ is
\begin{equation}
    G_a^1(p,n_a,k) =
    \begin{cases}
    (1-p)F + p\alpha^{\rm crit}(p,r,n_a,k)Q,& {\rm if}\ n_a< r,\\
    \alpha^{\rm crit}(p,r,n_a,k)Q, & {\rm otherwise}.
    \end{cases}
\end{equation}
\end{corollary}

\begin{theorem}
\label{thm.ign}
The strategy \critstrategy is unique equilibrium of \ifinesname{1}.
\end{theorem}
\begin{proof} 
 Consider $a\in\agentset$ in the sorted order. We will show by induction \critstrategy is a unique best-response to strategies of agents in front of a given agent. For the first agent, \critstrategy clearly maximizes the utility $-G_a$ of $a$. In the induction step we assume $a'$ in front of $a$ follow \critstrategy. Following Lemma~\ref{lem: irelevant after}, the actions of the others can be arbitrary. Observe the \critstrategy minimizes the expected payment (\ref{eq: exp 1 payment}). Thus $a$ wants to follow \critstrategy.
 \qed
\end{proof}

\subsubsection{More Sorting Instances}
\label{sub.ex}
In this section we present analytic solution of the general \ifinesname{w} game, $w\geq 1$. We start by defining extension of \critstrategy, and showing no agent can benefit by deviating from it. Later, we discuss some properties of this analytic solution.

In \ifinesname{w}, no agents are added after sorting. After the first \roundname the game is thus identical to \ifinesname{(w-1)}. This recursive relation motivates us to introduce the analogues of the variables used in the previous section recursively. We use upper index to denote the game length $w$ and number of \roundname, i.e. in the previous section we would use $r^{1,1}$ for the critical position $r$. 

\label{sub.ex}
We extend Definition~\ref{def.crit} of critical strategy to pay $F$ if $a$'s position is in front of some critical position $r^{w,t}$, defined below. Note that since the second \roundname corresponds to \ifinesname{(w-1)}, $r^{w,l} = r^{w-1,l-1}$ for $l>1$ and in particular $r^{w,w} = \dots = r^{2,2} = r^{1,1}=r$.
\begin{definition}[w-Critical strategy]
\label{def.cr2}
The {\it $w$-critical strategy} is 
\begin{equation}
\label{eq: t-crit strategy}
    \icritstrategy{w}(n_a, t_a, m_a) = 
    \begin{cases}
    \sigma^F & {\rm if}\ n_a < r^{w,t_a},\\
    \sigma^0 & {\rm otherwise}.
    \end{cases}
\end{equation}
\end{definition}
Let all agents follow $\icritstrategy{w}$. Then if $w>1$ and $a\in\agentset^1$ does not terminate in the first \roundname, his expected payment in the remaining $w-1$ rounds is
\begin{equation}
    \mathcal{G}_a^w(p, n_a, k) = \mathbb{E}_{\gamma\sim B(\min(n_a,r^{w,1})-1, 1-p)}[G^{w-1}_a(p, n_a-\gamma-k, k)],
\end{equation}
where $G^{w-1}_a$ is the recursive extension of the expected payment $G^1_a$ (see Corollary~\ref{cor: exp payment}). A formula for  $G^w_a$ is given in Proposition \ref{p:exp.payment} below.

In words, since all agents positioned in front of $\min(n_a,r^{w,1})$ want to pay $F$, $a$'s position decreases by $\gamma+k, \gamma\sim B(\min(n_a,r^{2,1})-1,1-p)$. At the new position, $a$ is expected to pay $G^{w-1}_a$.

\begin{proposition}
\label{p:exp.payment}
Let all agents follow \icritstrategy{w}, and $w>1$. Then the expected payment of an agent $a\in\agentset$ is
\begin{equation}
    G_a^w(p,n_a,k) =
    \begin{cases}
    (1-p)F + pX^w(p,r^{w,1},n_a,k),& {\rm if}\ n_a< r^{w,1},\\
    X^w(p,r^{w,1},n_a,k), & {\rm otherwise},
    \end{cases}
    \end{equation}
    where 
    $$
    X^w(p,r^{w,1},n_a,k)= 
    $$
    $$
    \alpha^{\rm crit}(p,r^{w,1},n_a,k)Q+ 
    (1-\alpha^{\rm crit}(p,r^{w,1},n_a,k))\mathcal{G}_a^w(p,n_a, k)
    $$
    is a's expected payment if he does not pay $F$ in the first Round.
\end{proposition}
It remains to determine critical positions $r^{w,l}$. Recursively, $r^{w,l}= r^{w-1,l-1}$ for $l> 1$. Hence it remains to define $r^{w,1}$. Similarly to Definition~\ref{def.crit}, we define the critical position in the first \roundname as the smallest $r^{w,1}\in\naturals$ such that $\alpha(p, r^{w,1}, k)Q + (1-\alpha(p, r^{w,1}, k))\mathcal{G}_a^w(p, r^{w,1}, k)\le F$. 

In words, assume all agents in front of $a$ want to pay $F$. In the first \roundname, if $a$ pays zero he risks paying $Q$ w.p. $\alpha$ and the expected payment in the remaining rounds w.p. $1-\alpha$. The critical position $r^{w,1}$ is the smallest position $n_a$ at which, assuming all agents in front of it try to pay $F$, it is beneficial to pay zero.

\begin{lemma}
\label{lemma: crit pos grow}
Let $w>1$. Then $r^{w,1} \ge r^{w-1,1} + k$.
\end{lemma}
\begin{proof}
By definition, $r^{w,1}$ is the smallest integer such that
$\alpha(p, r^{w,1}, k)Q + (1-\alpha(p, r^{w,1}, k))\mathcal{G}_a^w(p, r^{w,1}, k)\le F$. For a contradiction we assume that $r^{w,1} < r^{w-1,1} + k$. It suffices to show that
$\mathcal{G}_a^w(p, r^{w,1}, k) > F$ since this inequality along with $Q> F$ violates the defining property of $r^{w,1}$.

If $r^{w,1}-k < r^{w-1,1}$ then for each $\gamma\geq 0$, 
$$
G^{w-1}_a(p, r^{w,1} -\gamma-k, k)= (1-p)F + pX^{w-1}(p,r^{w-1,1},r^{w,1} -\gamma-k,k),
$$
see Proposition \ref{p:exp.payment}.
Moreover, by the defining property of $r^{w-1,1}$
$$
X^{w-1}(p,r^{w-1,1},r^{w,1} -\gamma-k,k)> F.
$$
Hence for each $\gamma\geq 0$, 
$$
G^{w-1}_a(p, r^{w,1} -\gamma-k, k)> F
$$
and thus $\mathcal{G}^w_a(p, r^{w,1}, k) > F$.
\qed
\end{proof}

We are now ready to show the main result of this section.

\begin{theorem}[Equilibrium of \ifinesname{w}]
\label{thm: w equilibrium} \icritstrategy{w} is unique equilibrium of \ifinesname{w}.
\end{theorem}
\begin{proof}
We proceed by induction on $w$. For $w=1$ we use Theorem~\ref{thm.ign}. After the first Round, the game corresponds to \ifinesname{(w-1)} and there is a unique equilibrium by the induction assumption. In the first \roundname, we can use a modification of proof of Theorem~\ref{thm.ign}: 
consider agents of $\agentset^1$ in the sorted order and use induction over agents. For an agent $a\in\agentset^1$ let his strategy be $\pi_a = (q, 1-q)$ in the first Round, he follows \icritstrategy{w} from the second Round, and let all agents in front of him follow \icritstrategy{w}. Then his expected payment is
\begin{equation}
\label{eq: rec exp payment}
(1-p-q)F + (p+q)X^w(p,r^{w,1},n_a,k),
\end{equation}

This is because w.p. $1-p-q$ he pays $F$ and leaves. Otherwise, since all agents in front of him follow \icritstrategy{w}, and he also follows \icritstrategy{w} from the second Round, his expected payment is $X^w(p,r^{w,1},n_a,k)$.

Strategy \icritstrategy{w} is chosen to minimize $a$'s expected payment (\ref{eq: rec exp payment}).
Therefore, $a$ will follow it even in the first \roundname.
\qed
\end{proof}

\begin{proposition}
\label{p.tot}
Let $w> 0$ be an integer and let all players follow \icritstrategy{w}. Then the total expected payment of \ifinesname{w} is 
\begin{equation}
\label{eq: 2 exp rev}
   wkQ+ F(1-p)\sum_{t=1}^w(r^{t,1}-1).
\end{equation}
\end{proposition}
\begin{proof}
In the first \roundname, $(1-p)(r^{w,1}-1)$ agents are expected to pay $F$, and $k$ are forced to pay $Q$. In the remaining rounds, the situation is analogous. 
\qed
\end{proof}

\begin{theorem}
\label{thm.avef}
Equilibrium strategies of all \ifinesname{w} exhibit the avalanche effect.
\end{theorem}
\begin{proof}
Since $\lim_{p\to0^+}\alpha(p,n,k)=1$ and $Q>F$, the critical position in the last \roundname $r^{w,w}\to\infty$. Using Lemma~\ref{lemma: crit pos grow}, the equilibrium strategies of all \ifinesname{w} satisfy $\icritstrategy{w}\to\sigma^F$ $\forall w\geq 1$. Thus, \icritstrategy{w} satisfies Definition~\ref{def: avalanche effect}.  
\qed
\end{proof}

In this simplified model, decreasing the probability of ignorance virtually increases the number of state employees assigned to processing the fines. 
This allows the central authority to increase the total payment through advertising, rather than hiring additional employees, which may be much cheaper.
We show in Section~\ref{sec: experiments} that these results translate well to a more general case where non-zero number of agents enter the system in each \roundname.

\subsubsection{Division problem}

To give a partial answer to the Division problem in this setting, we will compare the total expected payment of \ifinesname{w} with $k$, and \ifinesname{1} with $wk$.

\begin{theorem}
\label{thm.divv}
Let $Q\gg F$. Then the equilibrium strategy of \ifinesfull{w} achieves a higher total payment than the equilibrium of \ifinesfullk{1}{wk} in expectation by at least $F(1-p)w[k(w-1)-1]$.
\end{theorem}
\begin{proof}
By Proposition \ref{p.tot} and Lemma \ref{lemma: crit pos grow}, the expectation of the total payment of \ifinesfullk{w}{k} is at least 
$$
   wkQ+ F(1-p)\sum_{t=1}^w(r^{1,1}(k)-1)+ (t-1)k, 
$$
while the expectation of the total payment of \ifinesfullk{1}{wk} is  
$$
   wkQ+ F(1-p)(r^{1,1}(wk)-1). 
$$
To finish the proof, we note that by Proposition \ref{p.aa}, if $Q\gg F$ then $wr^{1,1}(k)\geq r^{1,1}(wk)$.
\qed
\end{proof}

XXXXXXXXXXXXXXXXXXXXXXX

\begin{definition}
\label{def.exppp}
For $m<n-k$, we let $\mathcal G_1^m(p,n,k)$ be the expectation of the payoff at position $n(m)$ of the Gangster game when all the players follow ${\mathcal R}_1$, where $n(m)$ is the random variable obtained from $n-k$ by $m$ independent events, each event consisting of subtracting $1$ with probability $1-p$. Hence, the expectation of $n(m)$ is $pm+n-k-m$.
\end{definition}

Next definition is important and quite recursive. For $n=1, \ldots$ we define inductively functions $m(p,n,k)$, $c_2(p,n,k)$, strategy ${\mathcal R}_2$ and critical position $n_2$. 

\begin{definition}
\label{def.2df}
For $n=k+1, \ldots$ we define inductively functions $m= m(p,n,k)$, $c_2(p,n,k)$, strategy ${\mathcal R}_2$ and critical position $r_2$ as follows.
\begin{itemize}
    \item $m(p,k+1,k)= 0$,
    \item 
    $
    c_2(p,n,k)= c_1(p,n,k)+ (1-\alpha'(p,n,k))\mathcal G_1^m(p,n,k),
    $
    \item
    If $c_2(p,n-1,k)> F$ then let $m= m(p,n,k)= n-k-1$,
    \item
    If $c_2(p,n-1,k)\leq F$ then let $r_2$ be smallest such that $c_2(p,r_2,k)\leq F$,
    \item
    If $c_2(p,n-1,k)\leq F$ then $m(p,n,k)= r_2-k-1$,
    \item
    the mixed strategy ${\mathcal R}_2$ for a player $a$ is: in the first round, pay $F$ with probability $1-p$ if $c_2(p,n_a,k)> F$ and pay $0$ otherwise. In the second round, follow mixed strategy ${\mathcal R}_1$.
    \end{itemize}
\end{definition}
\begin{lemma}
\label{l.comp}
$r_1+k\leq r_2.$
\end{lemma}
\begin{proof}
We first note that by definition, $m(p,r_2,k)= (r_2-k-1)$.

Abusing notation, $r_2$ is smallest such that $$c_1(r_2)+ (1-\alpha'(r_2))\mathcal G_1^{r_2-k-1}(r_2)\leq F.$$

If $r_2< r_1+k$ then $\mathcal G_1^{r_2-k-1}(r_2)> F$ by the definition of $r_1$, contradicting the defining inequality for $r_2$.

\end{proof}
\begin{proposition}
\label{p.aaa}
If $Q>>F$ then $r_1(2k)< 2r_1(k)$.
\end{proposition}
\begin{proof}
It suffices to show that $\alpha(p,r_1,k)\geq \alpha(p,2r_1,2k)$. If $r_1p \leq k$ then $\alpha(p,r_1,k)> F/Q$ contradicting the definition of $r_1$ and thus $r_1p > k$ and the statement follows from Proposition \ref{p.aa}.

\end{proof}

\begin{theorem}
\label{thm.p2}
Let $r_2$ be the critical position and let all players follow ${\mathcal R}_2$ in the 2-Fines game. The expected payoff $G_2(p,r_2,k)$ of player $a$ is: (1) If $n_a< r_2$ then $G_2(p,n_a,k)= (1-p)F+ pc_2(p,n_a,k)$, (2) If $r_2\geq n_a$ then $G_2(p,n_a,k)= c_2(p,n_a,k)$. 
The total expected payoff, denoted by $T_2(p,x_0,k)$ is
$
T_2(p,x_0,k)= (1-p)F(r_1+r_2-k-2) + 2Qk
$.
\end{theorem}

\begin{theorem}
\label{thm.ign}
The strategy of each player in each mixed equilibrium of the 2-Fines game is ${\mathcal R}_2$. Moreover, if $\mathbb F$ denotes the set of 2-Fines games where all players follow ${\mathcal R}_2$, then $\mathbb F$ exhibits the avalanche effect.
\end{theorem}
\begin{proof} 
The second statement follows from Theorem \ref{thm.p2}. We now show that ${\mathcal R}_2$ played at all positions of \agentset is a mixed equilibrium. Hence let $n$ be a position and let all positions of \agentset different from $n$ follow ${\mathcal R}_2$. We distinguish two cases:

1. $c_2(p,n,k)> F$. If position $n$ follows ${\mathcal R}_2$ then its expected payoff is $-(1-p)F- pc_2(p,n,k)$. If position $n$ violates ${\mathcal D}_2$ then it pays $F$ with probability $(1-p')$ for $p'> p$ and its expected payoff is $-(1-p')F- p'c_2(p,n,k)$ and thus it will decrease.

2.  Expected payoff of the critical position $r_2$ following ${\mathcal R}_2$ is equal to $-c_2(p,r_2,k)$. If $r_2$ does not follow ${\mathcal R}_2$ then it pays $F$ with some probability $0< 1-p'$ for $p'\geq p$ and thus its expected payoff will decrease to $-(1-p')F- p'Q\alpha(p,r_2,k)$.

3. For $n> r_2$, the expected payoff is bigger than $-c_2(p,r_2,k)$ so the same argument applies.

Finally, the uniqueness is proved as in the proof of Theorem \ref{thm.ign}.

\end{proof}
In $w$-Fines game, $k$ administrators settle fines
in $w$ rounds. Hence, they effectively need to perform $kw$ tasks. A natural question and ourwork in progress is, if we fix the number of tasks, what is the setting of $w$ and $k$ for maximising the total payoff? We report a very interesting partial answer to this basic question below. 
\begin{corollary}
\label{c.2fines}
If $Q>>F$ then the equilibrium strategy for $w=2, k$ give better total payoff than the equilibrium strategy for $w=1, 2k$.
\end{corollary}
\begin{proof}
This follows from Theorem \ref{thm.p}, Theorem \ref{thm.p2} and Proposition \ref{p.aaa}.
\end{proof}

\section{Experiments}
\label{sec: experiments}
We investigate two approaches based on how the agents choose their payments. In Section~\ref{ssec: brs}, we define a simple strategy based on how the agent's position changes over the course of the \gamename. In Section~\ref{ssec: ppo}, we use reinforcement learning to obtain a strategy which approximates equilibrium. In both cases we simplify the model by assuming the function $\pi_a$ is the same for all agents. The code is available at \href{https://github.com/DavidSych/Rule_Enforcing_Through_Ordering/tree/master}{GitHub}.

\subsection{Basic Rational Strategy}
\label{ssec: brs}
To model behaviour of real decision makers, we introduce {\it basic rational strategy} (BRS). Informally, each agent keeps track of a quantity he is willing to pay in each \roundname. If, based on his shift in \gamename since last \roundname, he determines he will reach the beginning before $T$ steps, his willingness to pay increases. Formally,
\begin{definition}[basic rational strategy]
Let $a\in\agentset$, $(n'_a, t'_a, m'_a)$ be the observation of $a$ in previous \roundname, and $(n_a, t_a, m_a)$ his current observation. We call $\omega_a$ the willingness to pay of $a$. In the first \roundname $a$ participates in, i.e. when $t_a = 0$, his willingness to pay is $\omega_a = 0$. In subsequent \roundname games, the willingness to pay is updated before declaring $\pi_a$ according to
\begin{equation}
    \omega_a \gets
    \begin{cases}
    \min(F-m_a, \omega_a + 1), & n_a < (n_a - n'_a)(T - t_a),\\
    \max(0, \omega_a - 1), & {\rm otherwise}.
    \end{cases}
\end{equation}
The strategy of $a$ is to pay $\omega_a$, i.e. $\pi_a = \sigma^{\omega_a}$.
\end{definition}
Note that this is a generalization of the approach introduced in Section~\ref{ssec: round}, as $\pi_a$ is not a function of only the observation in the current \roundname, but also depends on history. This makes this strategy non-Markovian. As such, the Definition~\ref{def: equilibrium} does not apply. However, in our experiments we simply assess the effect of agents using BRS, and make no claims regarding its optimality.

\subsection{Reinforcement Learning}
\label{ssec: ppo}
\begin{figure}[t]
\centering
\includegraphics[width=0.8\textwidth]{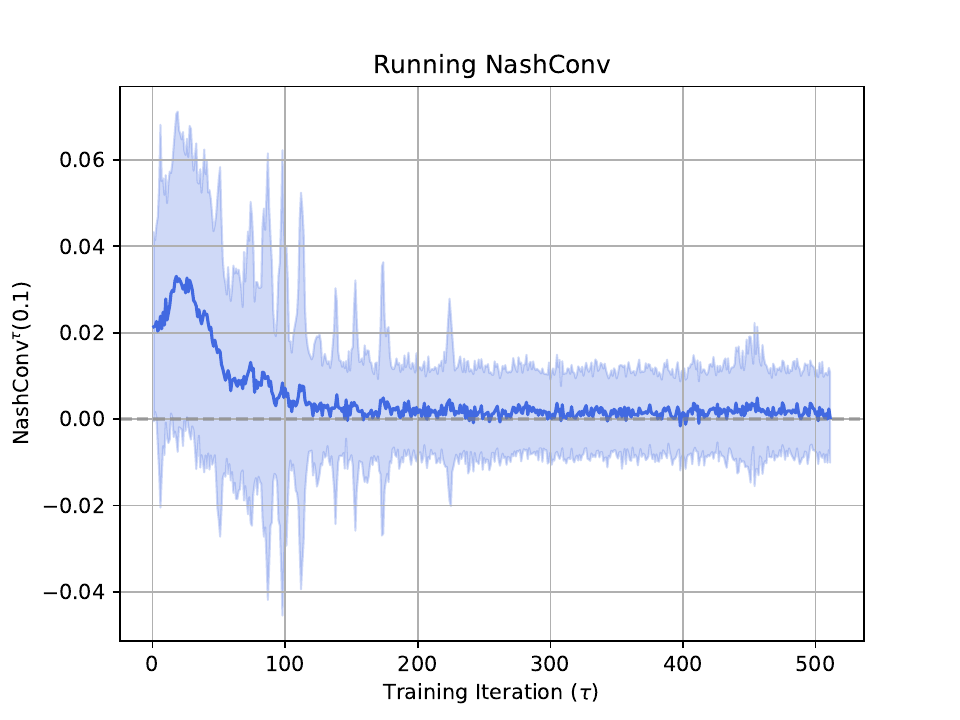}
\caption{Evolution of \nashconv during training, averaged over one hundred random seeds. The colored ares show standard error.}
\label{fig: nashconv}
\end{figure}
In order to approximate an equilibrium of \gamename, we employ an iterative algorithm. In each iteration, the algorithm approximates $\overline{\pi}_a$ such that
\begin{equation}
    \overline{\pi}_a \in 
    \argmax\ \mathbb{E}_{(\overline{\pi}_a, \pi_{-a})}\left[u_a(\overline{\pi}_a, \pi_{-a})\right].
\end{equation}
In words, we find $\overline{\pi}_a$ such that it maximizes utility of $a$, assuming $\agentset\setminus\{a\}$ follow $\pi$. We denote as $\tau$ the iteration of the learning algorithm and $\pi^\tau$ the strategy the algorithm approximates the best-response against in iteration $\tau$.

We use Proximal policy optimization (PPO) \cite{schulman2017ppo} to find $\overline{\pi}$, utilizing trajectories of all terminal agents for the update. For details on our implementation, see Appendix~\ref{app: ppo details}. This approach is not guaranteed to converge in general but if it does converge, the resulting strategy is an equilibrium~\cite{Ghosh2020Operator}. Similar approach was successfully used before~\cite{Baker2020Emergent}. 

\subsubsection{NashConv}
In order to quantify the quality of the learned solution, we adapt the notion of \nashconv \cite{NIPS2017_3323fe11}. \nashconv measures the negative difference in utility agents are expected to receive under $\pi^\tau$ and the approximate best-response $\pi^{\tau + 1}$. We approximate the latter by having a fraction of agents $\rho$ follow $\pi^{\tau + 1}$ while the rest follows $\pi^\tau$.
Formally,
\begin{definition}[NashConv]
Let each agent added to \gamename follow $\pi^{\tau+1}$ w.p. $\rho$ and $\pi^\tau$ otherwise. Let $\overline{\agentset}$ be the set of agents following $\pi^{\tau+1}$ and their expected utility
\begin{equation*}
    \mathcal{BRU}(\rho,\pi^{\tau+1}, \pi^\tau)=
    \mathbb{E}_{\left(\pi^{\tau+1}_{\overline{\agentset}}, \pi^\tau_{-\overline{\agentset}}\right)}\left[u_a(\pi^{\tau+1}_a, \pi^\tau_{-a})|a\in\overline{\agentset}\right].
\end{equation*}
Then
\begin{equation}
    \nashconv^\tau(\rho) = 
    \mathcal{BRU}(\rho,\pi^{\tau+1}, \pi^\tau) - 
    \mathbb{E}_{\pi^\tau}\left[u_a(\pi^\tau)\right].
\end{equation}
\end{definition}
\nashconv and $\epsilon$-equilibrium are closely connected -- if $\rho$ is small enough such that $|\overline{\agentset}| \ll |\agentset|$, then $\nashconv \approx \epsilon$.
In Figure~\ref{fig: nashconv} we present a representative example of the evolution of \nashconv during training. We averaged the results over one hundred random seeds, and also show the standard error. The results suggest that, although there is a considerable amount of noise, the algorithm was able to reach a sufficiently close approximation of the equilibrium. Moreover, we verified this trend translates to other experiments presented below.

\subsection{Results}
\label{ssec: results}
In this section, we numerically demonstrate the Avalanche effect and the Division problem. Specifically, we show the total expected revenue, which is given as 
$
    \mathbb{E}_{\pi}\left[\sum_{a\in\terminalset}m_a\right].
$
Unless stated otherwise, we use $F=T=4$, $Q=6$, $x=x_0=32$, $k=2$ and $p=1/2$ in all our experiments. Note that with these parameters if the ordering is not introduced\footnote{That is if the agents in $\agentset^t$ which are forced to pay $Q$ are selected at random.}, the individual risk in the first Round is $kQ / x = 0.375 \ll F$. Thus it is not rational to pay $F$ and the revenue of the central authority would be $kQ = 12$.

Note that the standard error is considerably high in all figures presented below.
This is partly due to the noise introduced by the learning algorithm, which (if convergent) find a course correlated equilibrium.
As these may vary significantly in e.g. social welfare, similar variance can be expected in our case.

\subsubsection{Avalanche Effect}
In Figure~\ref{fig: avalanche effect} we show the total expected payment as a function of the probability of ignorance $p$, and the number of entering agents $x$. The results suggest that the \gamename exhibits the Avalanche effect in a general setting. In fact, it exhibits both properties of Definition~\ref{def: avalanche effect}. Interestingly, the learned solution achieves a considerably higher total payment compared to BRS. 

\begin{figure*}[t!]
\centering
\begin{subfigure}{0.49\textwidth}
    \includegraphics[width=1\textwidth]{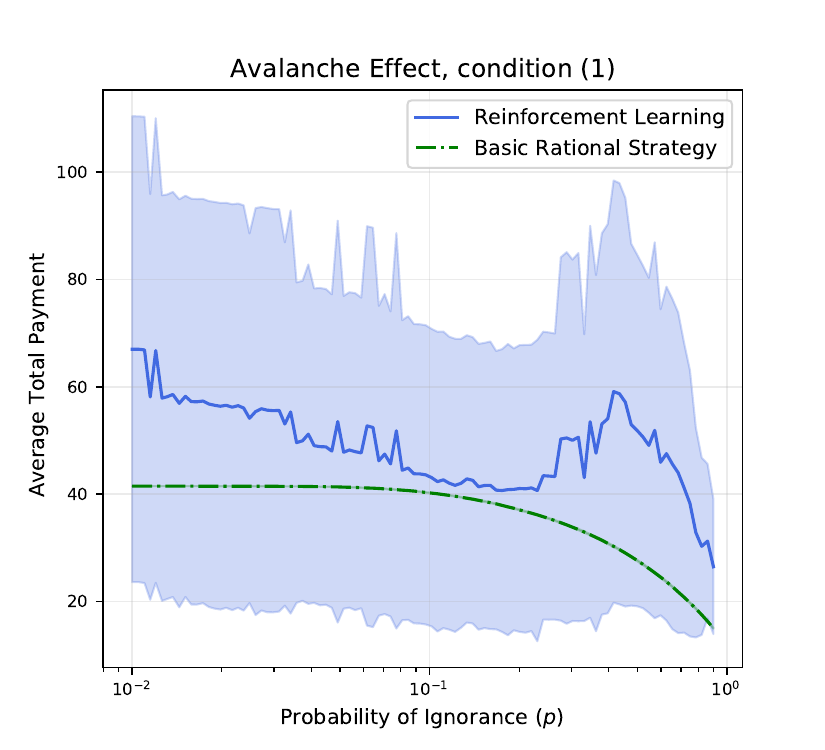}
\end{subfigure}
\hfill
\begin{subfigure}{0.49\textwidth}
    \includegraphics[width=1\textwidth]{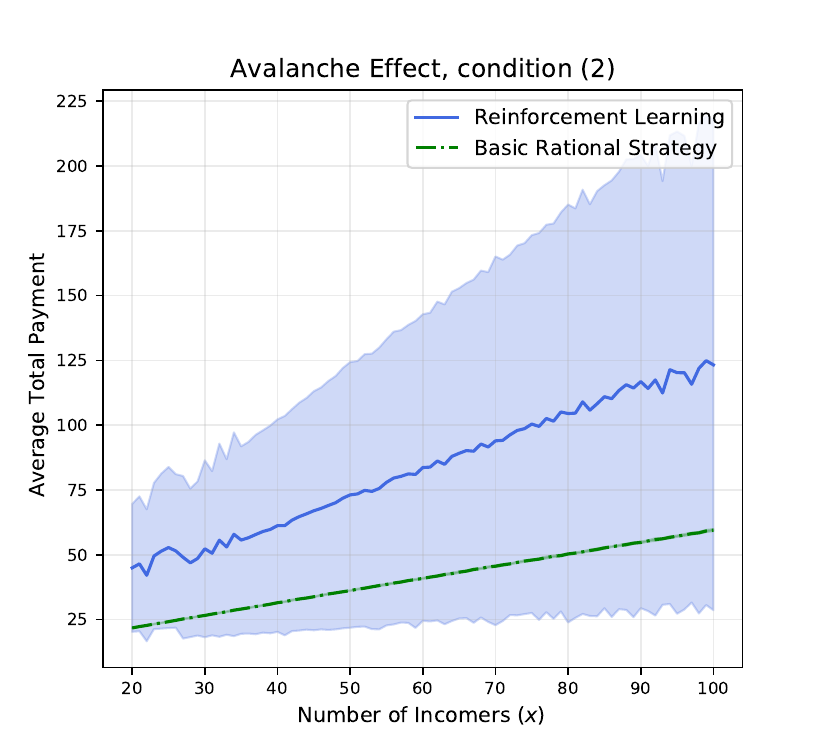}
\end{subfigure}
\caption{Expected total payment of terminal agents for varying probability of ignorance $p$ (left) and number of incoming offenders $x$ (right) averaged over ten random seeds and showing also the standard error. The figures demonstrate the Avalanche effect defined in Section~\ref{ssec: avelanche effect}.}
\label{fig: avalanche effect}
\end{figure*}

\subsubsection{Division problem}
\label{sssec: exp revenue}
In this section, we numerically study the Division problem introduced in Section~\ref{ssec: division problem}. Results for both the Time- and Group-Division problem are presented in Figure~\ref{fig: expected revenue}.

For the Time-Division problem, BRS seems to drastically overpay the learned strategy if the sorting is frequent, i.e. $T$ is large. On the other hand, when $T$ is small the willingness to pay doesn't increase. This leads to paying only $kQ=48$ for $T=1$, while the learned strategy prefers to pay more. When the game is sorted more often, the learned strategy seems to favor lower total payments.

In the Group-Division scenario, both BRS and the learned strategy pay less in larger system. Splitting the game into several smaller thus increases the total payment of the offenders. This is in agreement with the analytic solution presented in Section~\ref{sub.ex}, suggesting the incoming agents don't impact \gamename much.

\begin{figure*}[t!]
\centering
\begin{subfigure}{0.49\textwidth}
    \includegraphics[width=1\textwidth]{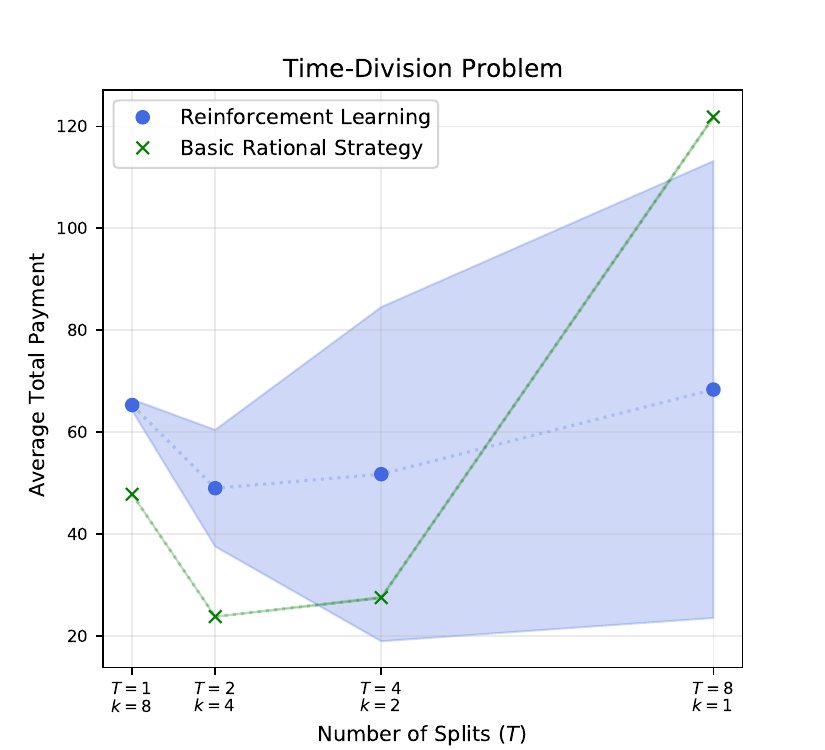}
\end{subfigure}
\hfill
\begin{subfigure}{0.49\textwidth}
    \includegraphics[width=1\textwidth]{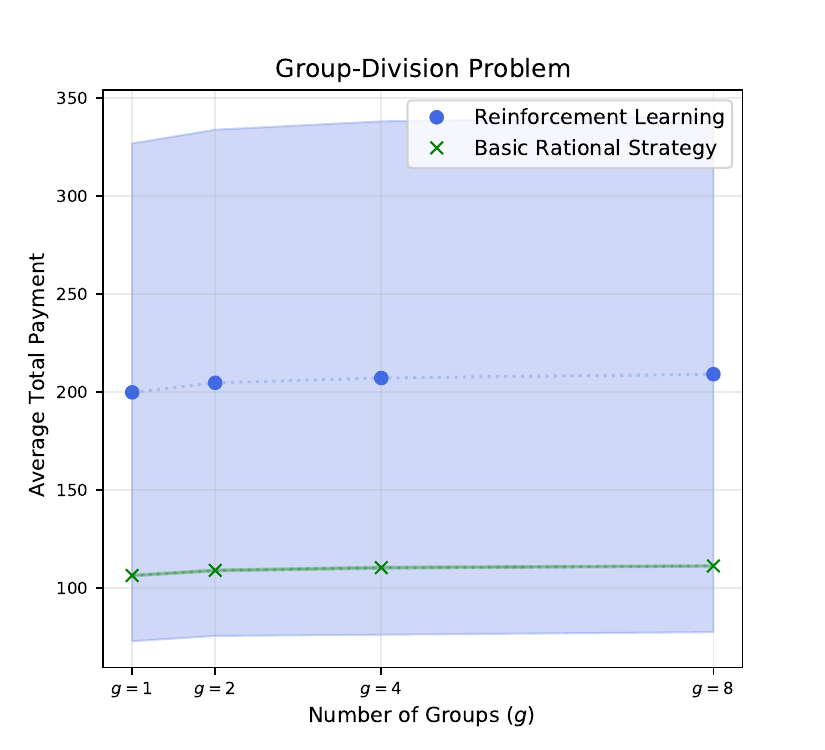}
\end{subfigure}
\caption{Expected total payment of terminal agents for varying number of sortings $T$ (left) and number of splits $g$ (right). The results are averaged over ten random seeds and the colored areas show standard error. The figures investigate the Division problem defined in Section~\ref{ssec: division problem}.}
\label{fig: expected revenue}
\end{figure*}

\subsubsection{Exploitability of Basic Rational Strategy}
The BRS is a heuristic designed to capture realistic behaviour of humans. However, it is not guaranteed to make optimal decisions. In this section, we investigate exploitability of BRS. Specifically, we let 90\% of the agents follow BRS, with the rest refining their strategy using PPO. We compare the expected payment of agents following each of the strategies after convergence. We present our results in Figure~\ref{fig: exploitability of brs} for varying probability of ignorance $p$ and number of entering agents $x$. In all cases the learning algorithm is able to find strategy which achieves vastly lower expected payment, suggesting the BRS is quite exploitable. 

\begin{figure*}[t!]
\centering
\begin{subfigure}{0.49\textwidth}
    \includegraphics[width=1\textwidth]{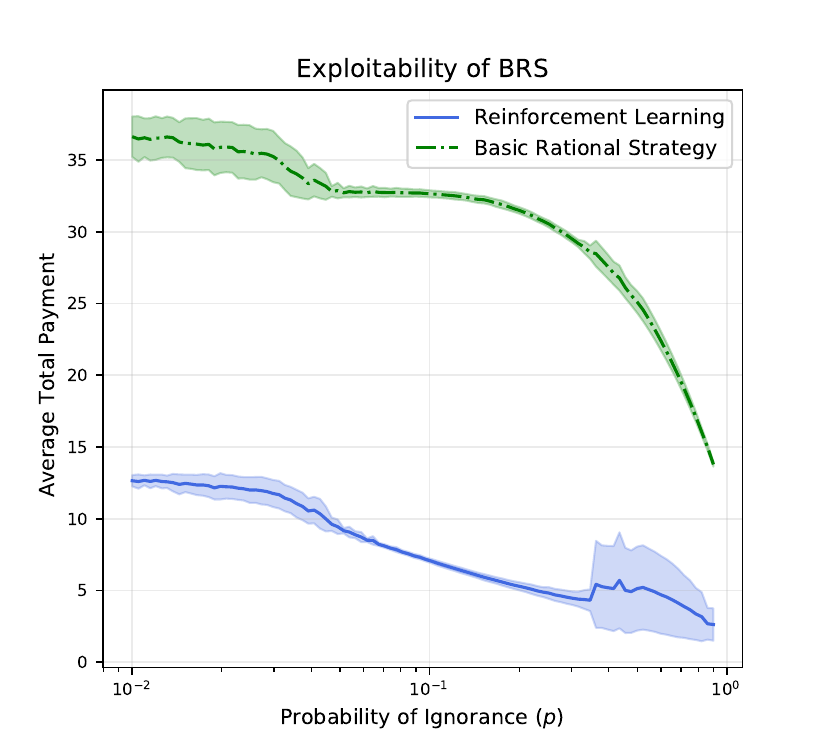}
\end{subfigure}
\hfill
\begin{subfigure}{0.49\textwidth}
    \includegraphics[width=1\textwidth]{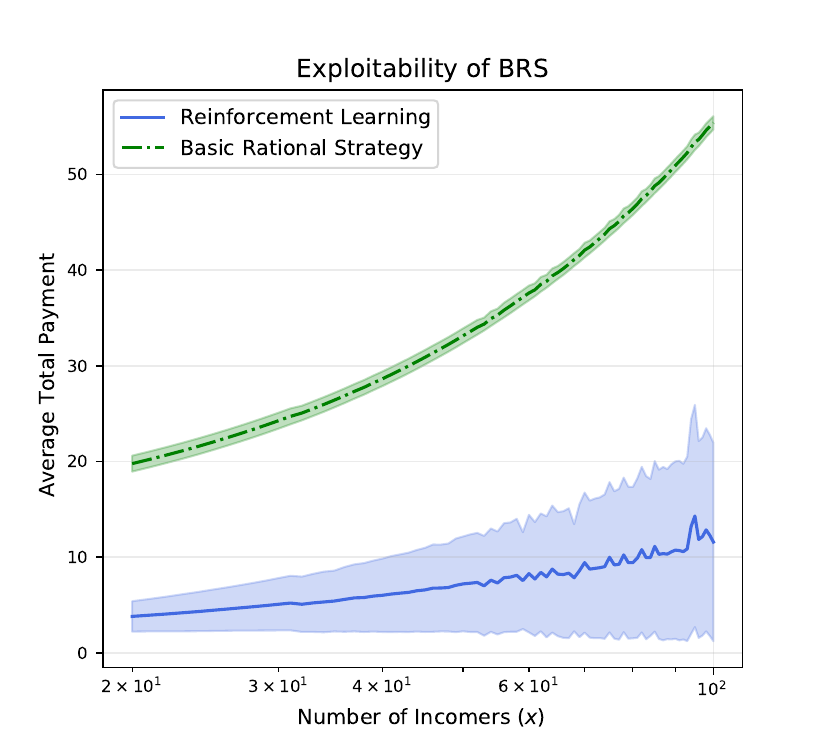}
\end{subfigure}
\caption{Expected total payment of terminal agents for varying probability of ignorance $p$ (left) and number of incoming offenders $x$ (right) averaged over ten random seeds and showing also the standard error. The training was done with 90\% of agents following BRS., i.e. approximating best-response to BRS.}
\label{fig: exploitability of brs}
\end{figure*}

\section{Conclusion}
\label{sec: conclusion}
In this work, we suggest a simple mechanism for rule enforcing, like collecting fines for traffic violations in large cities, by a small number of administrators. 
We show analytically and on realistic experiments that this simple mechanism exhibits the Avalanche effect and thus supports non-cooperation of offenders. We quantify the fines collection in expectation. Finally, we present some initial results towards understanding the effective use of the administrators, i.e., the Division problem.

\medskip

{\it Future work:} 
\*
Further study of the Division problem, in particular possible strengthening of Lemma \ref{lemma: crit pos grow}
is our work in progress. 

We see a limitation of our numerical approach in that we limit ourselves to scenarios where all agents share the same strategy $\pi_a$. We would like to improve on our results by having each agent follow one of a few leaders, similar to how we investigated exploitability of BRS.

\appendix

\section{Proof Of Theorem \ref{c.aa}}
\label{app: sameer proof}

\newtheorem{imr}{}
\newenvironment{cmr}[1]
  {\renewcommand\theimr{#1}\imr}
  {\endimr}
\begin{cmr}{Theorem \ref{c.aa}}
$\alpha(p,n,k)\geq \alpha(p,2n,2k)$ for $1 \leq k<np-p$. 
\end{cmr}

We will prove the theorem in a sequence of lemmas.
Note that $\alpha(p,n,k) = \mathbb{P}[X\leq k]$ for $X\sim B(n-1,p)$.

\begin{lemma}
\label{c.aa1}
For random variables $X\sim B(n,p)$ and $Y\sim B(2n,p)$ and $1< k<np$, 
we have $\mathbb{P}[X\leq k] \geq \mathbb{P}[Y\leq 2k]$.
\end{lemma}
\begin{proof}
We make use of the Camp-Paulson approximation~\cite{SLDJ,GREENBERG201491} to the normal distribution for a binomial distribution which states that for $X\sim B(n,p)$
$$\left\vert\mathbb{P}[X\leq k] - \mathbf{\Phi}\left( \frac{c-m}{\theta}\right)\right\vert \leq \frac{0.007}{\sqrt{np(1-p)}},$$ where $c=(1-b)r^\frac{1}{3},m=1-a,\theta=\sqrt{br^\frac{2}{3}+a},$
$b=\frac{1}{9(k+1)},a=\frac{1}{9(n-k)},r=\frac{(k+1)(1-p)}{p(n-k)}, $
and $\mathbf{\Phi}(x)=\frac{1}{\sqrt{2\pi}}{{\int}^x_{-\infty}}e^{-\frac{t^2}{2}}dt$.

Since $\mathbf{\Phi}$ is an increasing function it suffices to show the inequality between the arguments of $\mathbf{\Phi}$ for $k<np$.
We define $r(n,x)=\frac{(x+1)(1-p)}{p(n-x)}, c(n,x)=\left(1-\frac{1}{9(x+1)}\right)r(n,x)^\frac{1}{3} = \frac{9x+8}{9(x+1)}r(n,x)^\frac{1}{3}$, $m(n,x)=1-\frac{1}{9(n-x)}$ and $\theta(n,x)=\sqrt{\frac{1}{9(x+1)}r(n,x)^{2/3} +\frac{1}{9(n-x)}}$.

Thus we need to show that $\frac{c(n,x) - m(n,x)}{\theta(n,x)} > \frac{c(2n,2x) - m(2n,2x)}{\theta(2n,2x)}$ for $k < np$.
We prove this in two parts. Our first claim will show that there is a $K_n< np$, where ${c(n,x) - m(n,x)}$ is zero.
\qed
\end{proof}

\begin{claim}
$c(n,x)-m(n,x)$ is an increasing function of $x$ for $0< x<n$ and
there exists $K_n<np$ such that $c(n,x)<m(n,x)$ for all $x<K$ and $c(n,x)>m(n,x)$ for all $x>K$. 
\end{claim}
\begin{proof}
It is easy to see that for $0<x<n$, $r(x)$ and $c(x)$ are increasing functions and $m(n,x)$ is a decreasing function. Thus for $0<x<np$ we have $1> m(n,x)\geq 1-\frac{1}{9(n-np)}$ and $(1-\frac{1}{9(x+1)})\leq (1-\frac{1}{9(np+1)})$.
We first find the condition for $x>0$ such that $r(n,x)<\left( \frac{y-1}{y} \right)^3$ for some $y>0$. Note here that we can assume that such an $x$ exists as we are assuming $p>\frac{1}{n}$. The inequality holds for all $x  <  \frac{np(y-1)^3-y^3(1-p)}{y^3(1-p)+p(y-1)^3} $. Since $y>0$, we have that the inequality holds for all $x< np\left(\frac{y-1}{y}\right)^3 -1 + p$. Thus for $y=9(n-np)$ we have, $c(n,x)=\frac{9x+8}{9(x+1)}\left(1-\frac{1}{9(n-np)}\right) < m(n,x)$.

\remove{
\begin{eqnarray*}
\frac{(x+1)(1-p)}{p(n-x)} & < & \left( \frac{y-1}{y} \right)^3 \\
x & < & \frac{np(y-1)^3-y^3(1-p)}{y^3(1-p)+p(y-1)^3} \\
\end{eqnarray*}

Since $y>0$, we have $\frac{np(y-1)^3-y^3(1-p)}{y^3(1-p)+p(y-1)^3} > \frac{np(y-1)^3-y^3(1-p)}{y^3}$ and the inequality holds for all $x< np\left(\frac{y-1}{y}\right)^3 -1 + p$.
Thus for $y=9(n-np)$ we have for all $x< np\left(\frac{y-1}{y}\right)^3 -1 + p$, $c(n,x)=\frac{9x+8}{9(x+1)}\left(1-\frac{1}{9(n-np)}\right) < m(n,x)$.
}
\begin{eqnarray*}
c(n,np) & = & \frac{9np+8}{9(np+1)}\left(\frac{(np+1)(1-p)}{p(n-np)}\right)^{1/3}=\frac{9np+8}{9(np+1)}\left(\frac{np+1}{np}\right)^{1/3}\\
&\geq& \frac{9np+8}{9(np+1)}\left(\frac{np+1}{np}\right)^{1/3} = \frac{9np+8}{9np}\left(\frac{np}{np+1}\right)^{2/3} \\
& = &  \left(1+\frac{8}{9np}\right)\left(\frac{np}{np+1}\right)^{2/3} \\ 
\end{eqnarray*}
It is easy to see that $\left(1+\frac{8}{9x}\right)\left(\frac{x}{x+1}\right)^{2/3} >1$ for all $x>0$.
Thus $c(n,np) > 1 > m(n,np)$.
This proves the claim.
\qed
\end{proof}
Notice that $K_n$ is very close to $np$ but nevertheless lower than $np$.
We are now ready to partly prove Theorem~\ref{c.aa}. 
\begin{lemma}
\label{lem:kn}
For $0<x<\frac{K_{2n}}{2}$, $\frac{c(n,x) - m(n,x)}{\theta(n,x)} > \frac{c(2n,2x) - m(2n,2x)}{\theta(2n,2x)}$.
\end{lemma}
\begin{proof}

To do this we see some properties of $\frac{c(n,x) - m(n,x)}{\theta(n,x)}$.
Individually the functions compare as follows for $1\leq x<n$.

\begin{align*}
\left(\frac{\theta(2n,2x)}{\theta(n,x)}\right)^2 
& = \frac{1}{2}\left(\frac{x+1}{2x+1}\right)^{1/3} \frac{{(2n-2x)^{1/3}(1-p)^{2/3}+(2x+1)^{1/3}p^{2/3}}} {{(n-x)^{1/3}(1-p)^{2/3}+(x+1)^{1/3}p^{2/3}}} \\
& \leq  \frac{1}{2}\left(\frac{x+1}{2x+1}\right)^{1/3} \frac{{(2n-2x)^{1/3}(1-p)^{2/3}+(2x+2)^{1/3}p^{2/3}}} {{(n-x)^{1/3}(1-p)^{2/3}+(x+1)^{1/3}p^{2/3}}}\\
& \leq  \frac{1}{2^{2/3}} \left(\frac{x+1}{2x+1}\right)^{1/3} < 1
\end{align*}


Also $\frac{c(n,x)}{c(2n,2x)}  = 2^{1/3}\left(\frac{9x+8}{18x+8}\right)\left(\frac{2x+1}{x+1}\right)^{2/3} >1$ as this is a decreasing function for $x>0$ with its limit at $1$, and $m(n,x) - m(2n,2x) = \frac{1}{9(2n-2x)} - \frac{1}{9(n-x)} = -\frac{1}{9(2n-2x)} < 0$.


Thus we have $ c(2n,2x) - m(2n,2x) < c(n,x) - m(n,x) $.
It follows that $K_n \leq \frac{K_{2n}}{2}$.
Thus for $x\leq K_n$ we have $  \frac{\theta(2n,2x)}{\theta(n,x)} \frac{c(n,x) - m(n,x)}{c(2n,2x) - m(2n,2x)}<1$ i.e.,\\ $\left\vert\frac{c(2n,2x) - m(2n,2x)}{\theta(2n,2x)}\right\vert \geq \left\vert\frac{c(n,x) - m(n,x)}{\theta(n,x)}\right\vert$ but both quantities are negative and so $\frac{c(2n,2x) - m(2n,2x)}{\theta(2n,2x)} \leq \frac{c(n,x) - m(n,x)}{\theta(n,x)}$. For $K_n< x < \frac{K_2n}{2}$ we have $\frac{c(2n,2x) - m(2n,2x)}{\theta(2n,2x)} \leq 0 \leq \frac{c(n,x) - m(n,x)}{\theta(n,x)}$.
\qed 
\end{proof}




Lemma~\ref{lem:kn} allows us to state a weaker result.

\begin{corollary}
For random variables $X\sim B(n,p)$ and $Y\sim B(n+\left\lceil n/p\right\rceil,p)$ and $k<\max \{\frac{n}{2},np\}$, we have $\mathbb{P}[X\leq k] \geq \mathbb{P}[Y\leq 2k]$.
\end{corollary}
\begin{proof}
The proof follows from the fact that $n+\frac{n}{p} > 2n$ and $2x < (np+n)\left(\frac{9\frac{n}{p} -9np-1}{9\frac{n}{p}-9np}\right)-1+p<K_{n+\frac{n}{p}}$.
\qed
\end{proof}

Now we can complete the proof of Theorem~\ref{c.aa}.
\begin{proof}[of Theorem~\ref{c.aa}]
\allowdisplaybreaks
Notice that $c-m$ and $\theta$ are monotonically increasing in $x$. The difference between using $n$ and $2n$ is just the rate of increase. We have shown for $x< K_{2n}$, $(c-m)(n,x)\theta^2(2n,2x) > (c-m)(2n,2x)\theta^2(n,x)$. Now we show the inequality holds for $x=np$, i.e., the two functions haven't crossed each other. 

Define $r_1=\frac{np+1}{np},r_2=\frac{2np+1}{2np},b_1=\frac{1}{9(np+1)},b_2=\frac{1}{9(2np+1)}$,  $a=\frac{1}{18(n-np)}$, $\theta_1=b_1r_1^{2/3}+2a$ and $\theta_2=b_2r_2^{2/3}+2a$.
Thus we have 
\begin{equation}
1\leq \frac{r_1}{r_2} = 2\left(\frac{np+1}{2np+1}\right) =  \frac{2b_2}{b_1} \leq 2    
\end{equation}


\begin{align*}
&(c-m)(n,np)\theta^2(2n,2np)-(c-m)(2n,2np)\theta^2(n,np)  \\&=((1-b_1)r_1^{1/3} -1+2a)\theta_2-((1-b_2)r_2^{1/3} -1+a)\theta_1 &\\
& = \frac{2^{1/3}(2np+1)^{1/3}(9np+8)(n-np)-2^{2/3}(np+1)^{1/3}(18np+8)(n-np)}{81[(np+1)(2np+1)]^{2/3}2np(n-np)}\\&\quad +\frac{2[2np(np+1)]^{2/3}-2[np(2np+1)]^{2/3}+18(np+1)2^{1/3}[np(2np+1)]^{2/3}]^{2/3}}{81[(2np+1)(np+1)]^{2/3}(2np)(n-np)2} \\&\quad -\frac{18(2np+1)[2np(np+1)]^{2/3}}{81[(2np+1)(np+1)]^{2/3}(2np)(n-np)2}\\&\quad+\frac{18(n-np)[np(np+1)]^{1/3}(2np+1)^{2/3}-9(n-np)[2np(2np+1)]^{1/3}(np+1)^{2/3}}{81(n-np)[(np+1)(2np+1)]^{2/3}(2np)}\\ &\quad  +\frac{9[(np+1)(2np+1)]^{2/3}(2np)+2(np+1)^{2/3}(2np+1)^{1/3}(2np)^{1/3}}{81(n-np)[(np+1)(2np+1)]^{2/3}(2np)2} \\&\quad -\frac{2(np+1)^{1/3}(2np+1)^{2/3}(np)^{1/3}}{81(n-np)[(np+1)(2np+1)]^{2/3}(2np)2}\\
\end{align*}
Using $p^3-q^3=(p-q)(p^2+pq+q^2)$ we have,
\small{
\begin{align}
\label{eq:part1}
    & 2^{1/3}(2x+1)^{1/3}(9x+8)-9[2x(2x+1)]^{1/3}(x+1)^{2/3}\nonumber \\& =2^{1/3}(2x+1)^{1/3}[8+9x^{1/3}(x^{2/3}-(x+1)^{2/3}]\nonumber\\
    &=2^{1/3}(2x+1)^{1/3}\left[8+\left(\frac{-9x^{1/3}(2x+1)}{x^{4/3}+x^{2/3}(x+1)^{2/3}+(x+1)^{4/3}}\right)\right]
\end{align}
}
and,
\small{
\begin{align}
    & 18[x(x+1)]^{1/3}(2x+1)^{2/3}-2^{2/3}(x+1)^{1/3}(18x+8)\nonumber\\
    & =(x+1)^{1/3}[8+18x^{1/3}((2x+1)^{2/3}-(2x)^{2/3}]\nonumber\\
    &=(x+1)^{1/3}\left[8+\left(\frac{18x^{1/3}(4x+1)}{(2x+1)^{4/3}+(2x(2x+1))^{2/3}+(2x)^{4/3}}\right)\right]\label{eq:part2}
\end{align}
}
Note that the sum of~\ref{eq:part1}~and~\ref{eq:part2} is positive for $x\geq 1$. Thus all the terms with $n-np$ in the numerator add up to a positive quantity. The only other negative component is $\frac{18(np+1)2^{1/3}[np(2np+1)]^{2/3}-18(2np+1)[2np(np+1)]^{2/3}}{81[(2np+1)(np+1)]^{2/3}(2np)(n-np)2}$, which is dominated by $\frac{9[(np+1)(2np+1)]^{2/3}(2np)}{81(n-np)[(np+1)(2np+1)]^{2/3}(2np)2}$.

Thus $\frac{c(n,np) - m(n,np)}{\theta(n,np)}\Bigg/ \frac{c(2n,2np) - m(2n,2np)}{\theta(2n,2np)} \geq \frac{\theta(n,np)}{\theta(2n,2np)} \geq 1.$

\qed
\end{proof}

%

\begin{table}[h!]
\centering
\begin{tabular}{ |c | c | c| }
\hline
Parameter & Value & Description \\ 
\hline
\hline
$\varepsilon$ & 0.05 & Policy update clipping\\
\hline
$\gamma$ & 1 & Reward discounting\\
\hline
$\lambda$ & 0.95 & Advantage decay factor\\
\hline
$N_{\rm train}$ & 32 & Number of training updates per cycle\\
\hline
$N_{\rm epochs}$ & 512 & Number of training epochs\\
\hline
$N_{\rm train}$ & $2\cdot 10^{4}$ & Train buffer size\\
\hline
$\alpha_{\rm actor}$ & $3\cdot 10^{-4}$ & Actor learning rate\\
\hline
$\alpha_{\rm critic}$ & $10^{-3}$ & Critic learning rate\\
\hline
$c_{\rm H}$ & $10^{-3}$ & Entropy regularization weight\\
\hline
$\overline{c}$ & 0.1 & Gradient norm clipping\\
\hline
\end{tabular}
\caption{Hyperparameters of the learning algorithm.}
\label{tab: hyperparameters}
\end{table}

\section{Learning Algorithm}
\label{app: ppo details}
The shared strategy $\pi_a$ is represented by a neural network and trained from trajectories of all terminal agents. When selecting the strategy for a \roundname, we mask all actions which would lead to $m_a + \mu_a > F$. This makes the agents unable to overpay the fine $F$.
We use fully-connected networks for both the actor and the critic. Both take as input the observation\footnote{We normalize the observation to $[0, 1]^3$.} of $a$ in \roundname, i.e. $(n_a, t_a, m_a)$. The actor network has two hidden layers with four hidden units, and the critic has three hidden layers with 32 units each, all using the ReLU activation function. The rest of the hyperparameters are given in Table~\ref{tab: hyperparameters}.





\bibliographystyle{unsrt} 
\bibliography{main}


\end{document}
